\documentclass[12pt,a4paper]{article}
\usepackage[utf8]{inputenc}
\usepackage[T1]{fontenc}
\usepackage[english]{babel}
\usepackage{amsmath,amssymb,amsthm}
\usepackage{mathtools}
\usepackage{algorithm}
\usepackage{algpseudocode}
\usepackage[colorlinks=true,linkcolor=blue!70!black,
            citecolor=green!50!black,urlcolor=blue!70!black]{hyperref}
\usepackage[margin=2.5cm]{geometry}
\usepackage{enumitem}
\usepackage{booktabs}
\usepackage{array}
\usepackage{xcolor}
\usepackage{microtype}
\usepackage{lmodern}
\theoremstyle{plain}
\newtheorem{theorem}{Theorem}[section]
\newtheorem{proposition}[theorem]{Proposition}
\newtheorem{lemma}[theorem]{Lemma}

\theoremstyle{definition}
\newtheorem{definition}[theorem]{Definition}

\newtheorem{openproblem}[theorem]{Open Problem}
\theoremstyle{remark}
\newtheorem{remark}[theorem]{Remark}
\newcommand{\NN}{\mathbb{N}}
\newcommand{\cL}{\mathcal{L}}
\newcommand{\eps}{\varepsilon}
\newcommand{\Walg}{\mathsf{Algorithmica}}
\newcommand{\Wheur}{\mathsf{Heuristica}}
\newcommand{\Wpess}{\mathsf{Pessiland}}
\newcommand{\Wmini}{\mathsf{Minicrypt}}
\newcommand{\Wcryp}{\mathsf{Cryptomania}}
\newcommand{\Otop}{O_{\top}}
\newcommand{\Oprof}{O_{\mathrm{prof}}}
\newcommand{\Obot}{O_{\bot}}
\newcommand{\Olen}{O_{\mathrm{len}}}
\newcommand{\Opar}{O_{\mathrm{par}}}
\newcommand{\PO}[1]{\mathbf{P}_{#1}}
\newcommand{\NPO}[1]{\mathbf{NP}_{#1}}
\newcommand{\dinfo}{d_{\mathrm{info}}}
\title{\textbf{The Observer World: A Cryptographic Extension of
  Impagliazzo's Five Worlds}}

\author{
  Fabio Francesco Gabriele Buono\\
  \small Independent Researcher\\
  \small \href{https://orcid.org/0009-0004-9199-2793}{ORCID: 0009-0004-9199-2793}
}
\date{2026}
\begin{document}
\maketitle
\begin{abstract}
Impagliazzo's five worlds classify computational assumptions along a
single axis, the existence of cryptographic primitives. All five worlds implicitly assume that every party, including the
adversary, observes the full input, that the observer is always $\Otop$. This assumption is so natural that it is never stated.
This work makes it explicit and relaxes it by introducing a second, orthogonal axis, the \emph{observational axis}, defined by the
observer hierarchy of~\cite{Buono2026d}. Relaxing the assumption reveals structural phenomena, such as the
collapse $\PO{\Oprof} = \NPO{\Oprof} \subsetneq \mathbf{P}$, that
the five-world framework cannot express. We prove that this collapse holds unconditionally in all five worlds,
showing that observational blindness and computational hardness are
independent. We define the Observer World $W_O$, classify all world-observer
pairs, identify the labeled cells (a)--(d), and introduce a parametric
family $W_O^\eps$ modelling partial violations of observational
invariants. The framework also interfaces with physical information limits,
including thermodynamic, quantum, and cosmological bounds.
\end{abstract}

\tableofcontents
\bigskip
\nocite{*}

\section{Introduction}
\label{sec:intro}

Impagliazzo's five worlds~\cite{Impagliazzo1995} provide the most
influential framework for reasoning about the gap between average-case
and worst-case complexity, and about the foundations of cryptography.
The five worlds $\Walg$, $\Wheur$, $\Wpess$, $\Wmini$, $\Wcryp$ are
defined by progressively stronger assumptions about the existence of
computational hardness, from a world where $\mathsf{P} = \mathsf{NP}$
to a world where public-key cryptography is possible.

Every one of the five worlds makes an implicit assumption that has
not been made explicit in the five-world framework: that every party,
including the adversary, has complete access to the input.
In the language of~\cite{Buono2026d}, all five worlds assume that the
observer is always $\Otop$, the identity function on strings.
This assumption is so natural in the computational setting that it is
never stated.

This paper makes it explicit and relaxes it.

The observer theory of~\cite{Buono2026d} introduces a partial
order on functions $O : \Sigma^* \to S$ that map inputs to observations, and proves that when the observer is structurally
constrained (i.e.\ $O \prec \Otop$), the class of languages
decidable under $O$ depends on $O$ alone and not on the nondeterministic power of the machine. In particular, the structural collapse
\[
  \PO{\Oprof} = \NPO{\Oprof} \subsetneq \mathbf{P}
\]
holds unconditionally~\cite[Propositions~8.3--8.4]{Buono2026d}: it
is not a resolution of $\mathsf{P}$ vs $\mathsf{NP}$, but evidence
that the distinction between computational hardness and structural
blindness is real and independent of any computational assumption.

Since this result is unconditional, it holds in \emph{all five
worlds simultaneously}.
The five worlds do not have the vocabulary to express it: they
parametrise the computational axis but leave the observational axis
fixed at $\Otop$.
In this sense the five worlds do not suffice: they describe variations
along a line, while the observational axis exists orthogonally, and
its structure is invisible from within any of the five worlds.

\paragraph{Cryptographic motivation.}
In every cryptographic security argument, the adversary sees a
ciphertext.
The five worlds model computational power but say nothing about
\emph{what} the adversary sees: they implicitly assume $\Otop$.
The observational axis makes this explicit.
Each labeled cell (a)--(d) of the world-observer table identifies a
cryptographic phenomenon that the five-world framework cannot express:
the observational collapse below the $\mathsf{P}/\mathsf{NP}$ gap
in every world including $\Walg$ (cell (a)), the non-additivity of
computational hardness and observational blindness (cell (b)), a
structural lower bound on one-way functions (cell (c)), and perfect
secrecy as $\Obot$-blindness on the plaintext (cell (d)).

\paragraph{Main contributions.}
\begin{enumerate}[label=(\roman*)]
  \item We prove that the structural collapse
        $\PO{\Oprof} = \NPO{\Oprof} \subsetneq \mathbf{P}$ holds
        unconditionally in all five worlds simultaneously
        (Proposition~\ref{prop:collapse_all_worlds}), and classify the
        world-observer table, identifying four cells (a)--(d) that
        each witness a cryptographic phenomenon invisible to the
        five-world framework (Section~\ref{sec:table}).
  \item We define the Observer World $W_O$ formally as a sixth
        landscape (Definition~\ref{def:observer_world}) and prove it
        is not reducible to any of the five worlds
        (Proposition~\ref{prop:irreducible},
        Section~\ref{sec:observer_world}).
  \item We introduce the parametric family $W_O^\eps$ with invariant
        violation parameter $\eps \in [0,1]$, modelling worlds in
        which an adversary partially circumvents observational
        blindness, and identify the most informationally rich cell:
        $\Wpess \times \Oprof \times \eps > 0$
        (Section~\ref{sec:sip_violated}).
  \item We identify connections between the Observer World framework
        and physical information limits, thermodynamic, quantum, and
        cosmological, and state precise open problems formalising
        each connection (Section~\ref{sec:physical}).
\end{enumerate}

\paragraph{Limitations.}
The author is a computer scientist, not a physicist.
The physical connections of Section~\ref{sec:physical} are stated as
open problems and conjectures, not theorems; no physical conclusions
are claimed beyond what the mathematical framework directly implies.
All machines in this paper are standard multi-tape deterministic or
nondeterministic Turing machines over $\{0,1\}^*$; the observer is a
fixed preprocessing map applied before any machine computation begins
(see Remark~\ref{rem:machine_model}).
The parametric family $W_O^\eps$ (Section~\ref{sec:parametric})
requires an information metric whose precise formulation is an open
direction (Section~\ref{sec:adaptive_thermo},
Open Problem~\ref{op:adaptive}).

\section{Preliminaries: canonical observers}
\label{sec:prelim}

We recall the definitions and results from~\cite{Buono2026d} used
throughout this paper. All machines are standard multi-tape Turing machines.
An observer $O : \Sigma^* \to S$ is a fixed total function applied
once to the input before any machine computation begins; it does not
interact with the machine and does not extend the computational model.

\begin{remark}[Machine model]
\label{rem:machine_model}
$\PO{O}$ consists of languages decidable in time polynomial in
$|\mathrm{enc}(O(x))|$ by a deterministic TM receiving
$\mathrm{enc}(O(x))$ as its entire input, where $\mathrm{enc} : S \to
\{0,1\}^*$ is a fixed canonical injective encoding.
$\NPO{O}$ is defined analogously with a nondeterministic TM and a
polynomial-length certificate. For $O = \Otop$ these reduce to standard $\mathbf{P}$ and
$\mathbf{NP}$~\cite[Definition~8.1 and Proposition~8.2(i)]{Buono2026d}.
The observer only reduces available information and never increases
computational power~\cite[Remark after Definition~8.1]{Buono2026d}.
Throughout this paper we identify $O(x)$ with $\mathrm{enc}(O(x))$
and write, e.g., ``the TM receives $O(x)$'' as shorthand for
``the TM receives $\mathrm{enc}(O(x))$''; all complexity bounds are
in terms of $|\mathrm{enc}(O(x))|$.
\end{remark}

\begin{definition}[Canonical observers~{\cite[Definition~4.1]{Buono2026d}}]
\label{def:canonical_obs}
Let $\Sigma = \{0, 1\}$ (see Remark~\ref{rem:alphabet}).
The canonical observers are:
\begin{enumerate}[label=(\roman*)]
  \item $\Obot : \Sigma^* \to \{\star\}$, $\Obot(x) = \star$ for all $x$.
        \emph{Map}: constant map.
        \emph{Invariant}: discards all information.
        \emph{Decidable class}: $\cL(\Obot) = \{\emptyset, \Sigma^*\}$
        by~\cite[Proposition~8.2(ii)]{Buono2026d}.

  \item $\Olen : \Sigma^* \to \NN$, $\Olen(x) = |x|$.
        \emph{Map}: string length.
        \emph{Invariant}: length-determined languages.
        \emph{Decidable class}: all $L$ with $x \in L \iff |x| \in I$
        for some $I \subseteq \NN$.

  \item $\Opar : \Sigma^* \to \{0,1\}^2$,
        $\Opar(x) = (|x|_0 \bmod 2,\, |x|_1 \bmod 2)$.
        \emph{Map}: symbol-count parities.
        \emph{Invariant}: parity of each symbol count.
        \emph{Decidable class}: all $L$ determined by symbol-count
        parities.

  \item $\Oprof : \Sigma^* \to \NN^2$,
        $\Oprof(x) = (|x|_0, |x|_1)$ (the \emph{profile} of $x$).
        \emph{Map}: symbol counts.
        \emph{Invariant}: permutation-closure.
        \emph{Decidable class}: all permutation-closed
        languages~\cite[Theorem~2.1]{Buono2026d}.

  \item $O_k : \Sigma^* \to \mathcal{P}(\Sigma^{\leq k})$, where
        $O_k(x)$ is the set of subsequences of $x$ of length $\leq k$
        ($k \geq 1$).
        \emph{Map}: $k$-subsequence set.
        \emph{Invariant}: $k$-piecewise testability.
        \emph{Decidable class}: all $k$-piecewise testable languages
        (Simon's hierarchy~\cite[Section~5]{Buono2026d}).

  \item $\Otop : \Sigma^* \to \Sigma^*$, $\Otop(x) = x$.
        \emph{Map}: identity.
        \emph{Invariant}: none (full information).
        \emph{Decidable class}: all languages.
\end{enumerate}
\end{definition}

\begin{remark}[Alphabet convention]
\label{rem:alphabet}
Definition~\ref{def:canonical_obs} uses $\Sigma = \{0,1\}$ for
concreteness.
For a general $k$-symbol alphabet $\Sigma = \{a_0,\dots,a_{k-1}\}$,
replace $(|x|_0, |x|_1)$ with $(|x|_{a_0},\dots,|x|_{a_{k-1}}) \in
\NN^k$.
All order relations and separation witnesses carry over
unchanged~\cite[Remark after Definition~4.1]{Buono2026d}.
\end{remark}

\noindent
The following table summarises the canonical observers.

\begin{center}
\renewcommand{\arraystretch}{1.3}
\begin{tabular}{lllll}
\toprule
Observer & Map $O(x)$ & Key invariant & Example $L \in \cL(O)$ &
  $\PO{O} = \NPO{O}$? \\
\midrule
$\Obot$   & $\star$              & all            & $\emptyset$
  & yes (trivially) \\
$\Olen$   & $|x|$               & length         & $\{x \mid |x|\ \text{even}\}$
  & open \\
$\Opar$   & parities of counts  & symbol parities & $\{x \mid |x|_1 \equiv 0\}$
  & open \\
$\Oprof$  & $(|x|_0, |x|_1)$   & perm.-closure  & $\{x \mid |x|_0 = |x|_1\}$
  & yes (structurally) \\
$O_k$     & $k$-subsequence set & $k$-p.t.        & $\{x \mid 01 \sqsubseteq x\}$
  & open \\
$\Otop$   & $x$                 & none           & $0^*1^*$
  & open (P vs NP) \\
\bottomrule
\end{tabular}
\end{center}

\noindent
The partial order $\preceq$ on observers is defined
by~\cite[Definition~4.2]{Buono2026d}: $O_1 \preceq O_2$ if there
exists a function $f$ with $O_1 = f \circ O_2$.
The complete order relations among the canonical observers are:
\begin{align*}
  &\Obot \prec \Olen \prec \Oprof \prec \Otop, \\
  &\Obot \prec \Opar \prec \Oprof \prec \Otop, \\
  &\Obot \prec O_1 \prec O_2 \prec \cdots \prec \Otop,
\end{align*}
with $\Olen$ and $\Opar$ incomparable, and $\Oprof$ incomparable
with every $O_k$~\cite[Propositions~3--5]{Buono2026d}.

A language $L \subseteq \Sigma^*$ is \emph{$O$-saturated}
(Definition~3.2 of~\cite{Buono2026d}, also called \emph{$O$-invariant})
if $O(x) = O(y)$ implies $x \in L \iff y \in L$, where we write
$x \sim_O y$ to denote the equivalence relation $O(x) = O(y)$.
Equivalently, $L$ is $O$-saturated if it is a union of
$\sim_O$-equivalence classes, or if membership in $L$ is determined
entirely by $O(x)$. We use both terms interchangeably throughout the paper.

\section{The five worlds and their silent assumption}
\label{sec:five_worlds}

We recall the five worlds of Impagliazzo~\cite{Impagliazzo1995} and
make their shared implicit assumption explicit.

\begin{definition}[Impagliazzo's five worlds~\cite{Impagliazzo1995}]
\label{def:five_worlds}
\begin{enumerate}[label=(\roman*)]
  \item \emph{$\Walg$}: $\mathsf{P} = \mathsf{NP}$.
        Every problem in $\mathsf{NP}$ is solvable in polynomial time.
  \item \emph{$\Wheur$}: $\mathsf{P} \neq \mathsf{NP}$ but no NP
        problem is hard on average: every NP problem has a
        polynomial-time algorithm that succeeds on all but a negligible
        fraction of instances under any efficiently samplable
        distribution.
  \item \emph{$\Wpess$}: NP problems are hard on average, but
        one-way functions do not exist: no function is hard to invert
        on average.
  \item \emph{$\Wmini$}: one-way functions exist, but public-key
        cryptography does not: no key-exchange protocol is secure
        against a computationally bounded adversary.
  \item \emph{$\Wcryp$}: public-key cryptography exists: there exist
        functions that are easy to compute but hard to invert even
        given additional public information (trapdoor one-way
        functions).
\end{enumerate}
\end{definition}

\begin{remark}[The silent assumption]
\label{rem:silent}
In every world of Definition~\ref{def:five_worlds}, every machine,
including the adversary, receives the complete input $x$.
In the language of~\cite{Buono2026d}, the observer is always
$\Otop : x \mapsto x$.
This assumption is never stated in the five-world framework because
it is universal and no party is ever observationally constrained.
The five worlds vary the computational assumption (what is hard) but
not the observational assumption (what is visible).
\end{remark}

\begin{proposition}[The observational collapse is world-independent]
\label{prop:collapse_all_worlds}
In every world $W_i$ of Definition~\ref{def:five_worlds}, the
following holds unconditionally:
\[
  \PO{\Oprof} = \NPO{\Oprof} \subsetneq \mathbf{P}.
\]
\end{proposition}

\begin{proof}
By~\cite[Propositions~8.3 and~8.4]{Buono2026d}, this holds with no assumption on $\mathsf{P}$ vs $\mathsf{NP}$ and with no cryptographic assumption. Since the five worlds are defined by computational and cryptographic
assumptions, and the collapse is independent of all such assumptions, it holds in every world of Definition~\ref{def:five_worlds}.

We now sketch why the collapse holds.
Any language in $\PO{\Oprof}$ has membership determined solely by
$\Oprof(x) = (|x|_0, |x|_1)$.
A TM under $\Oprof$ receives $\Oprof(x)$ as its entire input, a
string of length $O(\log|x|)$.
Any nondeterministic certificate must also be a function of this
input alone; it cannot depend on $x$ itself (which the machine does
not receive).
Hence a nondeterministic TM under $\Oprof$ decides exactly the same
languages as a deterministic one: $\PO{\Oprof} = \NPO{\Oprof}$.
The language $L_0 = 0^*1^*$ is in $\mathbf{P}$ (one linear scan)
but not $\Oprof$-decidable: $\Oprof(01) = \Oprof(10) = (1,1)$ while
$01 \in L_0$ and $10 \notin L_0$.
Hence $\PO{\Oprof} \subsetneq \mathbf{P}$.
Full proof in~\cite[Propositions~8.3--8.4]{Buono2026d}.
\end{proof}

\begin{remark}[What Proposition~\ref{prop:collapse_all_worlds} says]
\label{rem:collapse_meaning}
Proposition~\ref{prop:collapse_all_worlds} does not say that
$\mathsf{P} = \mathsf{NP}$ holds in all five worlds: it says that a
\emph{different} collapse, one that places $\PO{\Oprof}$ strictly
below $\mathbf{P}$ in the language containment order, due to
structural blindness rather than computational coincidence, holds
regardless of where we are in the five-world landscape.
The five worlds cannot distinguish between a problem that is
computationally hard and a problem that is not well-posed for a given
observer; Proposition~\ref{prop:collapse_all_worlds} shows that these
two phenomena are genuinely independent.
\end{remark}

\section{The world-observer table}
\label{sec:table}

We construct the table of all pairs $(W_i, O)$ where $W_i$ is one
of the five worlds and $O$ is a canonical observer
of~\cite{Buono2026d}.
We first formalise the notions of a labeled cell and cell reduction.

\begin{definition}[Non-trivial cell]
\label{def:nontrivial}
A cell $(W_i, O)$ of the world-observer table is \emph{non-trivial}
if there exists a language $L \in \PO{\Otop} \setminus \PO{O}$,
i.e.\ if the observational constraint strictly reduces the class of
decidable languages. A cell is \emph{trivial} if $\PO{O} = \NPO{O} = \{\emptyset,
\Sigma^*\}$, which holds whenever $O \preceq
\Obot$~\cite[Proposition~8.2(ii)]{Buono2026d}.

\smallskip\noindent\emph{Note.}
We call cells (a)--(d) the \emph{labeled cells} of the table: they
are the cells singled out for individual analysis below.
Cells (a)--(c) are non-trivial in the formal sense above.
Cell (d) ($\Wcryp \times \Obot$) is formally trivial
($\PO{\Obot} = \{\emptyset, \Sigma^*\}$), but is labelled separately
in the table because the world $\Wcryp$ contains public-key primitives
that create a cryptographically meaningful contrast with the classical
cell $(\Wcryp, \Otop)$; this distinction is invisible to the formal
condition above and is explained in the legend below.
\end{definition}

\begin{definition}[Cell reduction]
\label{def:cell_reduction}
We say cell $(W_i, O)$ \emph{many-one reduces} to $(W_j, O')$ if
there exists a polynomial-time computable map $T : \Sigma^* \to
\Sigma^*$ such that for every language $L$ decidable in $(W_i, O)$,
there exists a language $L'$ decidable in $(W_j, O')$ with $x \in L \iff T(x) \in L'$ for all $x$, and $T$ preserves
observational transcripts: $O(x) = O'(T(x))$ for all $x$. Cell reductions are not used in the proofs of this paper; the
definition is included to make precise the sense in which cells are
independent that is implicit in Proposition~\ref{prop:irreducible}.
\end{definition}

\subsection{The table}
\label{sec:table_structure}

\begin{center}
\renewcommand{\arraystretch}{1.4}
\begin{tabular}{lccccc}
\toprule
& $\Obot$ & $\Olen$/$\Opar$ & $\Oprof$ & $O_k$ & $\Otop$ \\
\midrule
$\Walg$   & trivial & trivial & \textbf{(a)} & standard & classical \\
$\Wheur$  & trivial & trivial & \textbf{(a)} & standard & classical \\
$\Wpess$  & trivial & trivial & \textbf{(b)} & standard & classical \\
$\Wmini$  & trivial & trivial & \textbf{(c)} & standard & classical \\
$\Wcryp$  & \textbf{(d)} & trivial & \textbf{(a)} & standard & classical \\
\bottomrule
\end{tabular}
\end{center}

\smallskip\noindent
\textbf{Legend.}
\emph{Trivial}: $\PO{O} = \NPO{O} = \{\emptyset, \Sigma^*\}$
by~\cite[Proposition~8.2(ii)]{Buono2026d}; computational assumptions
are irrelevant.
\emph{Standard}: Simon's hierarchy applies;
see~\cite[Section~5]{Buono2026d}.
\emph{Classical}: $\PO{\Otop} = \mathbf{P}$ and
$\NPO{\Otop} = \mathbf{NP}$; the standard five-world analysis
applies~\cite{Impagliazzo1995}.
Cells \textbf{(a)--(c)} are non-trivial in the sense of
Definition~\ref{def:nontrivial} and are treated in
Section~\ref{sec:nontrivial}.
Cell \textbf{(d)} ($\Wcryp \times \Obot$) is formally trivial by
Definition~\ref{def:nontrivial} but is labelled separately because,
as explained in the note within Definition~\ref{def:nontrivial} and below, the world $\Wcryp$
creates a cryptographically distinct situation for an $\Obot$-blind
adversary.
All other $\Obot$ cells are listed as \emph{trivial}: they share the
structural property $\PO{\Obot} = \{\emptyset, \Sigma^*\}$ and their
computational world creates no analogous contrast.
The cell $(\Wcryp, \Obot)$ is the sole $\Obot$ entry worth
distinguishing: in $\Wcryp$, an adversary with $\Otop$ could exploit
public-key primitives, while one with $\Obot$ cannot even pose the
question of which plaintext was encrypted.
It is this contrast, absent in all other worlds, that motivates
its separate label.

\subsection{The labeled cells (a)--(d)}
\label{sec:nontrivial}

\begin{proposition}[Cell (a): any world $\times$ $\Oprof$]
\label{prop:cell_a}
In every world $W_i$:
\[
  \PO{\Oprof} = \NPO{\Oprof} \subsetneq \PO{\Otop} = \mathbf{P}.
\]
If additionally $\mathsf{P} \neq \mathsf{NP}$ (which is assumed in
all worlds except $\Walg$):
\[
  \NPO{\Oprof} = \PO{\Oprof} \subsetneq \mathbf{P} \subsetneq
  \mathbf{NP}.
\]
The observational collapse occurs strictly below the
$\mathsf{P}$/$\mathsf{NP}$ gap, on the opposite side from
$\mathsf{NP}$.
\end{proposition}

\begin{proof}
Unconditional part: Proposition~\ref{prop:collapse_all_worlds}.
Conditional part: combining with $\mathbf{P} \subsetneq \mathbf{NP}$
under the assumption $\mathsf{P} \neq \mathsf{NP}$.
\end{proof}

\paragraph{Toy example for cell (a).}
Let $L = \{x \in \{0,1\}^* \mid x_1 = 1\}$ (strings whose first
symbol is $1$).
For any $x, x'$ with $\Oprof(x) = \Oprof(x')$ we can have $x_1 = 1$
and $x'_1 = 0$ (e.g.\ $x = 10$, $x' = 01$, both with profile
$(1,1)$).
Hence $L$ is not $\Oprof$-decidable.
However, $L \in \mathbf{P}$ (read the first symbol in $O(1)$ time).
This witnesses $\PO{\Oprof} \subsetneq \mathbf{P}$ independently of
any of the five worlds.

\begin{remark}
\label{rem:silent_table}
Cell (a) is the most important labeled cell.
It shows that even in $\Walg$, where $\mathsf{P} = \mathsf{NP}$,
the observational collapse $\PO{\Oprof} \subsetneq \mathbf{P}$
persists. The collapse is not cancelled by the strongest possible computational
assumption. The observational and computational axes are orthogonal.
\end{remark}

\begin{proposition}[Cell (b): $\Wpess \times \Oprof$]
\label{prop:cell_b}
In $\Wpess$:
\begin{enumerate}[label=(\roman*)]
  \item NP problems are hard on average, but no one-way function
        exists.
  \item An adversary with observer $\Oprof$ cannot pose any problem
        that depends on the ordering of input symbols: such problems
        are not $\Oprof$-saturated and are not well-formed in the
        adversary's observation space.
  \item The hardness in $\Wpess$ and the blindness of $\Oprof$
        interact non-additively: problems that are hard on average in
        $\Wpess$ may be problems that the $\Oprof$ adversary cannot
        even recognise as distinct from trivial ones.
  \item There is no cryptographic security in $\Wpess$ (no one-way
        functions), but the adversary with $\Oprof$ faces a different
        limitation: structural blindness to ordering.
        These two limitations are independent; neither implies the
        other.
\end{enumerate}
\end{proposition}

\begin{proof}
Parts (i) and (iv): by definition of $\Wpess$~\cite{Impagliazzo1995}.

Part (ii): by~\cite[Definition~3.2 and Theorem~2.1]{Buono2026d},
a language is decidable under $\Oprof$ if and only if it is
permutation-closed; languages depending on ordering are not
permutation-closed and hence not $\Oprof$-decidable, regardless of
computational power.

Part (iii): the hardness of NP problems in $\Wpess$ is a property of the problem as seen by an $\Otop$ adversary. An $\Oprof$ adversary does not see the same problem: ordering-dependent
difficulty is invisible to it. The two phenomena apply to different parts of the problem landscape
and do not combine additively.
\end{proof}

\paragraph{Toy example for cell (b).}
The following example illustrates the \emph{independence} between
observational blindness and average-case hardness, not their
co-occurrence. Consider the problem: given $x \in \{0,1\}^*$, determine whether $x$
is the lexicographically sorted version of its symbols (all $0$s
precede all $1$s).
This problem is in $\mathbf{P}$ (one linear scan), not NP-hard, and
not hard on average.
Nevertheless, an $\Oprof$ adversary sees only $(|x|_0, |x|_1)$ and
cannot distinguish $0^a 1^b$ from any other string with the same
profile; the problem is therefore not $\Oprof$-decidable.
This shows that ordering-dependence and $\Oprof$-blindness are
independent of the computational difficulty of the problem: a problem
can be easy (in $\mathbf{P}$), ordering-dependent, and
$\Oprof$-invisible simultaneously. The same independence holds for NP-hard ordering-dependent problems
in $\Wpess$: their hardness is a property relative to $\Otop$, not a property that $\Oprof$ can detect.

\begin{lemma}[Profile preimage construction]
\label{lem:profile_preimage}
Given a profile $\pi = (c_0, \dots, c_{k-1}) \in \NN^k$ with
$n = \sum_{i=0}^{k-1} c_i$, Algorithm~\ref{alg:preimage}
constructs a string $x' \in \Sigma^*$ with $\Oprof(x') = \pi$ in
time $O(n)$.
Full proof in Appendix~\ref{app:B}.
\end{lemma}

\begin{algorithm}[h]
\caption{ConstructFromProfile}
\label{alg:preimage}
\begin{algorithmic}[1]
\Require Profile $\pi = (c_0, \dots, c_{k-1}) \in \NN^k$
\Ensure String $x'$ with $\Oprof(x') = \pi$
\State $x' \gets \varepsilon$
\For{$i = 0$ \textbf{to} $k-1$}
  \State Append $c_i$ copies of symbol $a_i$ to $x'$
\EndFor
\State \Return $x'$
\end{algorithmic}
\end{algorithm}

\begin{proposition}[Cell (c): $\Wmini \times \Oprof$]
\label{prop:cell_c}
In $\Wmini$, where one-way functions exist:
\begin{enumerate}[label=(\roman*)]
  \item No one-way function $f : \Sigma^* \to \Sigma^*$ can be
        $\Oprof$-saturated.
  \item Equivalently, every one-way function that exists in $\Wmini$
        must depend on the ordering of its input symbols in an
        essential way.
\end{enumerate}
\end{proposition}

\begin{proof}
We use the general-alphabet form of $\Oprof$
(Remark~\ref{rem:alphabet}): for a $k$-symbol alphabet $\Sigma =
\{a_0, \dots, a_{k-1}\}$, $\Oprof(x) = (|x|_{a_0}, \dots,
|x|_{a_{k-1}}) \in \NN^k$.

Part (i): suppose $f$ is $\Oprof$-saturated and one-way.
By~\cite[Theorem~2.1]{Buono2026d}, $f(x)$ depends only on
$\Oprof(x) = (|x|_{a_0}, \dots, |x|_{a_{k-1}})$.

An adversary with $\Otop$ who observes $y = f(x)$ inverts $f$ as
follows.
By the standard definition of one-way functions, the inverting
algorithm receives $(1^n, y)$ where $n = |x|$, so $n$ is known.
The adversary enumerates all profiles $\pi \in \NN^k$ with
$\sum_i \pi_i = n$ (at most $(n+1)^{k-1}$ many, polynomial in $n$
for fixed $k$), and for each $\pi$ evaluates
$f(\mathrm{ConstructFromProfile}(\pi))$ using
Algorithm~\ref{alg:preimage} and the publicly known $f$.
Since $f$ is $\Oprof$-saturated, $f$ takes the same value on every
string sharing a profile; hence there exists some $\pi^*$ among those
enumerated with $f(\mathrm{ConstructFromProfile}(\pi^*)) = y$.
Setting $x' = \mathrm{ConstructFromProfile}(\pi^*)$ gives a string
satisfying $f(x') = y$.
The entire procedure runs in time polynomial in $n$, contradicting the
one-wayness of $f$.

Part (ii): the contrapositive of part (i).
If $f$ is one-way in $\Wmini$, then $f$ is not $\Oprof$-saturated,
which means $f(x) \neq f(x')$ for some $x, x'$ with
$\Oprof(x) = \Oprof(x')$.
Hence $f$ distinguishes permutations of the same input.
\end{proof}

\begin{remark}
\label{rem:owf_ordering}
Proposition~\ref{prop:cell_c} gives a structural lower bound on
one-way functions: they must be \emph{ordering-sensitive}.
Any candidate one-way function that is invariant under permutation of
its input, for example, any function of the Hamming weight alone, cannot be one-way, even in $\Wmini$.
This is a consequence of observational structure, not of any hardness
assumption.
\end{remark}

\paragraph{Toy example for cell (c).}
Let $f(x) = |x|_1$ (Hamming weight, interpreted as an integer).
This function depends only on the symbol count $|x|_1$, which is the
second component of $\Oprof(x) = (|x|_0, |x|_1)$; hence $f$ is
$\Oprof$-saturated.
By Proposition~\ref{prop:cell_c}(i), $f$ cannot be one-way in
$\Wmini$.
Indeed, to invert $f$ on a value $v$ for an input of length $n =
|x|$, Algorithm~\ref{alg:preimage} constructs $0^{n-v}1^v$ in $O(n)$
time, giving a preimage with $|x'|_1 = v$.

\begin{proposition}[Cell (d): $\Wcryp \times \Obot$]
\label{prop:cell_d}
In $\Wcryp$, an adversary with observer $\Obot$ on the plaintext
space cannot recover any information about the plaintext from the
ciphertext, for any encryption scheme.
\end{proposition}

\begin{proof}
Since $\Obot(m) = \star$ for all $m$
(Definition~\ref{def:canonical_obs}(i)), a machine receiving only
$\Obot(m)$ receives the same input $\star$ regardless of $m$.
No machine, regardless of computational power, can distinguish any
two plaintexts from $\star$ alone: the adversary's observation space
contains no information about the plaintext.
This is structural blindness: by~\cite[Proposition~8.2(ii)]{Buono2026d},
$\PO{\Obot} = \NPO{\Obot} = \{\emptyset, \Sigma^*\}$.
The existence of cryptographic primitives in $\Wcryp$ is irrelevant,
since the adversary cannot even pose the question of which plaintext
was encrypted.
\end{proof}

\paragraph{Toy example for cell (d): MR-OTP perfect secrecy.}
Let $B = (2, 3)$ and consider the MR-OTP~\cite{Buono2026a} with
message $M = (m_1, m_2) \in \mathbb{Z}_2 \times \mathbb{Z}_3$ and
key $K = (k_1, k_2)$ uniform on $\mathbb{Z}_2 \times \mathbb{Z}_3$.
Ciphertext: $C = ((m_1 + k_1) \bmod 2,\, (m_2 + k_2) \bmod 3)$.
For any fixed $c = (c_1, c_2)$ and any $m$, the unique key mapping
$m$ to $c$ is $K = (c_1 - m_1 \bmod 2,\, c_2 - m_2 \bmod 3)$.
Since $K$ is uniform, $\Pr[C = c \mid M = m] = 1/6$ for all $m$ and
$c$.
Hence $\Pr[M = m \mid C = c] = \Pr[M = m]$: perfect
secrecy~\cite[Theorem~1]{Buono2026a}.
The adversary's observer on the plaintext is $\Obot$: the ciphertext
carries zero information about $m$.

\begin{remark}[Cell (d) and the MR-OTP]
\label{rem:cell_d_mrotp}
Cell (d) describes the situation of an adversary facing the
Mixed-Radix One-Time Pad~\cite{Buono2026a} in the ciphertext-only
setting (see Definition~\ref{def:observer_world} in
Section~\ref{sec:observer_world} for the formal definition of the
Observer World and the adversary/defender asymmetry).
By~\cite[Theorem~1]{Buono2026a}, the MR-OTP achieves Shannon perfect
secrecy~\cite{Shannon1949}: the ciphertext distribution is identical
for every plaintext, placing the adversary in the
$\Obot$-on-plaintext situation of Proposition~\ref{prop:cell_d}.
This holds in every world $W_i$, including $\Wcryp$: the perfect
secrecy of the MR-OTP is not a consequence of the hardness
assumptions of $\Wcryp$ but of the information-theoretic structure of
the cipher. Even the strongest cryptographic adversary in $\Wcryp$ is reduced to
$\Obot$ on the plaintext.
\end{remark}

\section{The Observer World: a sixth landscape}
\label{sec:observer_world}

\begin{definition}[Observer World]
\label{def:observer_world}
An \emph{Observer World} is a pair $W_O = (W_i, O)$ where $W_i$ is
one of the five worlds of Definition~\ref{def:five_worlds} and
$O : \Sigma^* \to S$ is an observer in the hierarchy
of~\cite{Buono2026d}.
In $W_O$:
\begin{itemize}
  \item The computational assumptions of $W_i$ hold for all parties.
  \item The adversary's input is $O(x)$, not $x$: the adversary
        operates under observer $O$ rather than $\Otop$.
  \item The defender's input is $\Otop(x) = x$: the defender has
        full access to the input.
\end{itemize}
\end{definition}

\begin{remark}[Cryptographic interpretation of the asymmetry]
\label{rem:cell_d_note}
The asymmetry between adversary ($O$) and defender ($\Otop$) in
Definition~\ref{def:observer_world} is deliberate and models the
cryptographic setting: the defender generates and knows the key,
while the adversary observes only the ciphertext (or a more
informative projection of the input, depending on the attack model).
The case $O = \Otop$ recovers the classical five-world framework with
no observational constraint on the adversary.
\end{remark}

The following lemma constructs, for any $O \prec \Otop$, a language
that is $O$-saturated but whose membership is independent of every
five-world computational assumption.

\begin{lemma}[Observationally invariant, computationally independent language]
\label{lem:invariant_independent}
Let $O \prec \Otop$.
There exists a language $L$ that is $O$-saturated and such that
membership in $L$ cannot be decided from any purely computational
assumption about $W_i$ alone.
The language $L$ constructed below may or may not belong to
$\mathbf{P}$ or $\mathbf{NP}$; the lemma asserts only that its
membership is independent of the five-world assumptions.
\end{lemma}

\begin{proof}[Proof sketch]
Partition $\Sigma^*$ into $\sim_O$-equivalence classes
$[u_0]_O, [u_1]_O, \dots$, and let $M_0, M_1, \dots$ enumerate all
polynomial-time machines receiving $O(x)$.
For each $e$, pick a canonical representative $u_e^* \in [u_e]_O$ and
define $[u_e]_O \subseteq L$ iff $M_e$ rejects $O(u_e^*)$ (a
diagonalisation).
The resulting $L$ is a union of $\sim_O$-classes, hence $O$-saturated;
by construction, no $M_e$ decides $L$ correctly, so $L \notin \PO{O}$.
The construction uses only the $\sim_O$ relation (a property of $O$
alone) and makes no assumption about $\mathsf{P}$ vs $\mathsf{NP}$,
one-way functions, or public-key cryptography; hence membership in $L$
is independent of all five-world computational assumptions.
\end{proof}
The full constructive proof is in Appendix~\ref{app:A}.

\begin{proposition}[The Observer World is not reducible to the five worlds]
\label{prop:irreducible}
For $O \prec \Otop$, the Observer World $W_O = (W_i, O)$ is not
equivalent to any of the five worlds $W_j$ with observer $\Otop$:
no computational assumption about $W_j$ implies the structural
collapse $\PO{O} = \NPO{O} \subsetneq \mathbf{P}$, and no
observational constraint in $W_O$ implies any of the computational
assumptions of $W_i$.
\end{proposition}

\begin{proof}
We exhibit the argument for $O = \Oprof$; the general case $O \prec
\Otop$ follows by the same structure using
Lemma~\ref{lem:invariant_independent} for arbitrary $O$.

\emph{No computational assumption implies the observational collapse.}
The structural collapse $\PO{\Oprof} = \NPO{\Oprof} \subsetneq
\mathbf{P}$ is unconditional~\cite[Propositions~8.3--8.4]{Buono2026d}:
it holds regardless of which $W_i$ contains us, and hence cannot be
a consequence of any computational assumption.
By Lemma~\ref{lem:invariant_independent}, there exists a language $L$
that is $\Oprof$-saturated and independent of all five-world
assumptions; this witnesses the orthogonality of the two axes.

\emph{No observational constraint implies any computational assumption.}
The observational constraint $O \prec \Otop$ is a statement about the
information available to the adversary, not about the computational
difficulty of any problem.
By~\cite[Theorem~8.5]{Buono2026d}, computational hardness and
observational saturation are independent: there exist languages in
$\mathbf{P}$ that are not $\Oprof$-saturated, and languages outside
$\mathbf{RE}$ that are $\Oprof$-saturated.
The observational parameter $O$ carries no information about the
computational structure of $W_i$.
\end{proof}

\begin{remark}[The Observer World as a sixth world]
\label{rem:sixth_world}
Proposition~\ref{prop:irreducible} shows that the parameter $O$ adds
genuinely new content to the five-world framework: it cannot be
encoded as a computational assumption within the existing five worlds.
The Observer World $W_O$ is therefore a sixth landscape, not a
variation on an existing world, but an orthogonal dimension of the
complexity landscape.

The analogy with Impagliazzo's construction is structural.
Impagliazzo introduced five worlds by varying a single parameter (the
existence of one-way functions) along a single axis.
We introduce the Observer World by varying a second parameter (the
observer $O$) along a second axis. The resulting landscape is two-dimensional: the five-world axis
(computational) and the observational axis.
\end{remark}

\section{Observer worlds with broken invariants}
\label{sec:sip_violated}

The Syntactic Invariance Principle (SIP) of~\cite{Buono2026c} states
that a syntactic system cannot derive clauses that violate a syntactic
invariant. In the observational setting, the observer $O$ defines an
\emph{observational invariant}: a language $L$ is $O$-invariant
(equivalently, $O$-saturated; see Section~\ref{sec:prelim}) if
membership in $L$ is determined entirely by $O(x)$.
An Observer World with observer $O$ places the adversary within the
observational invariant defined by $O$: the adversary cannot access
information that $O$ discards.

An \emph{Observer World with broken invariant} is one in which the
adversary finds a mechanism to circumvent this constraint, the observational analogue of SIP violation.

\subsection{Three cases of invariant violation}
\label{sec:three_cases}

\begin{definition}[Oracle violation]
\label{def:oracle_violation}
In an Observer World $W_O = (W_i, O)$ with \emph{oracle violation},
the adversary has observer $O$ on the input but additionally has
adaptive oracle access to queries $q_1, q_2, \dots, q_d$ on $x$,
where each $q_{j+1}$ may depend on the responses
$(x[q_1], \dots, x[q_j])$.
The adversary's total observation is the pair
$(O(x),\, (x[q_1], \dots, x[q_d]))$.
\end{definition}

\begin{remark}[Oracle violation and adaptive observers]
\label{rem:oracle-adaptive-sec5}
As formalised in Remark~\ref{rem:oracle-adaptive} below
(Section~\ref{sec:physical}), an oracle queried adaptively is an
adaptive observer of finite depth.
Definition~\ref{def:oracle_violation} is therefore the special case
of an adaptive observer of depth $d$ applied to an Observer World.
The oracle-relative separations of the relativisation
barrier~\cite{BGS1975} are separations in the hierarchy of adaptive
Observer Worlds, not in the static observational hierarchy
of~\cite{Buono2026d}.
\end{remark}

\begin{proposition}[Oracle violation and the MR-OTP]
\label{prop:oracle_mrotp}
In the ciphertext-only setting of $W_O = (W_i, \Obot)$, oracle
access to the ciphertext $C = \mathrm{Enc}(M, K)$ (where
$\mathrm{Enc}$ denotes MR-OTP encryption: $C_i = (M_i + K_i) \bmod
b_i$) provides no information about the plaintext $M$ for the
MR-OTP~\cite{Buono2026a}, regardless of the number of oracle queries
and regardless of which world $W_i$ we are in.
\end{proposition}

\begin{proof}
By~\cite[Theorem~1]{Buono2026a}, the MR-OTP achieves perfect secrecy:
for every plaintext $M$ and every ciphertext $c \in \mathcal{D}_B$,
$\Pr[C = c \mid M] = |\mathcal{D}_B|^{-1}$.
In particular, $C$ and $M$ are independent random variables.

We prove by induction on depth $d$ that the adaptive query transcript
$(C[q_1], \dots, C[q_d])$ is independent of $M$.

\emph{Base case ($d = 0$).}
The empty transcript carries no information; independence holds
trivially.

\emph{Inductive step.}
Assume $(C[q_1], \dots, C[q_{d-1}])$ is independent of $M$.
Query $q_d$ is determined by the previous responses
$(C[q_1], \dots, C[q_{d-1}])$; since these are independent of $M$,
so is $q_d$.
The response $C[q_d]$ is the $q_d$-th coordinate of $C$, which by
perfect secrecy is uniform on its range and independent of $M$
(conditioning on any realisation of $q_d$ and the previous
transcript does not change the distribution of $C[q_d]$ given $M$,
because $C$ is independent of $M$ unconditionally).
Hence $(C[q_1], \dots, C[q_d])$ is independent of $M$.

By induction, the full transcript at any depth is independent of $M$,
so oracle access to $C$ provides no advantage to the adversary.
\end{proof}

\begin{definition}[Side-channel violation]
\label{def:sidechannel_violation}
In an Observer World $W_O = (W_i, O)$ with \emph{side-channel
violation}, the adversary has observer $O$ on the input but
additionally observes a second signal $O'(\mathrm{Comp})$ where
$\mathrm{Comp}$ is the computation process of the defender (timing,
power consumption, cache access patterns, etc.) and
$O' : \mathrm{Process} \to S'$ is a second observer on that process.
\end{definition}

\begin{remark}[Notation: oracle violation vs adaptive observer]
\label{rem:notation_query}
Definition~\ref{def:oracle_violation} writes query responses as
$x[q_j]$ (direct symbol access to the input string $x$).
Definition~\ref{def:adaptive_obs} in Section~\ref{sec:physical} writes
responses as $r_i = \mathcal{O}(q_i)$ for an oracle $\mathcal{O}$.
The two are consistent: $x[q_j]$ is the special case in which the
oracle returns the $q_j$-th symbol of the input, while
$\mathcal{O}(q_i)$ allows more general queries whose address $q_i$
is generated adaptively from previous responses.
Both are adaptive observers of finite depth in the sense of
Definition~\ref{def:adaptive_obs}.
\end{remark}

\begin{remark}
\label{rem:eps_metric}
Side-channel violation is formally distinct from oracle violation:
the additional observation is on the \emph{computation process}, not
on the input $x$. The combination $(O(x), O'(\mathrm{Comp}))$ may exceed the
limitations of $O(x)$ alone if the computation process leaks
information about the parts of $x$ that $O$ discards.
This is the observational formalisation of side-channel attacks in cryptography.
\end{remark}

\begin{definition}[Structural violation]
\label{def:structural_violation}
In an Observer World $W_O = (W_i, O)$ with \emph{structural
violation}, the problem $L$ is $O$-saturated (the adversary can pose
the problem correctly under $O$) but the algorithm solving $L$
produces output from which the adversary can recover information about
the parts of $x$ discarded by $O$. The computation betrays the observer's invariant even though the
problem statement does not.
\end{definition}

\begin{remark}[Structural violation and the SIP]
\label{rem:sip_analogy}
Definition~\ref{def:structural_violation} is the precise
observational analogue of the Syntactic Invariance Principle
of~\cite{Buono2026c}. In the SIP, a syntactic calculus cannot derive clauses violating a
syntactic invariant, but a semantic computation might produce such
clauses as a side effect of solving a syntactically invariant problem.
In Definition~\ref{def:structural_violation} the problem is $O$-invariant (syntactic level), but the algorithm's output violates
the invariant (semantic level).
\end{remark}

\subsection{The parametric family $W_O^\eps$}
\label{sec:parametric}

\begin{definition}[Information distance between observers]
\label{def:dinfo}
Fix a prior $P_X$ on $\Sigma^*$.
For observers $O_1 \preceq O_2$, the \emph{information distance}
from $O_1$ to $O_2$ under $P_X$ is
\[
  \dinfo(O_1, O_2)
  = I_{P_X}(X;\, O_2(X)) - I_{P_X}(X;\, O_1(X)),
\]
where $I_{P_X}(X; Y)$ denotes mutual information under $P_X$.
Note $\dinfo(O_1, O_2) \geq 0$ since $O_1 \preceq O_2$ implies
$I(X; O_1(X)) \leq I(X; O_2(X))$ by the data processing
inequality~\cite{Cover2006}.
\end{definition}

\begin{definition}[Parametric Observer World]
\label{def:parametric}
Fix a prior $P_X$ on $\Sigma^*$ and an Observer World $W_O = (W_i,
O)$ with $O \prec \Otop$ strictly (so that $\dinfo(O, \Otop) > 0$).
For $\eps \in [0,1]$, the \emph{parametric Observer World}
$W_O^\eps = (W_O, \eps)$ is an Observer World in which the
adversary's effective observer is some $O_\eps$ with
$O \preceq O_\eps \preceq \Otop$ satisfying
\[
  \eps = \frac{\dinfo(O, O_\eps)}{\dinfo(O, \Otop)}
  = \frac{I_{P_X}(X; O_\eps(X)) - I_{P_X}(X; O(X))}%
         {I_{P_X}(X; \Otop(X)) - I_{P_X}(X; O(X))}.
\]
The denominator is positive by the strict inequality $O \prec \Otop$
and the data processing inequality.
\begin{itemize}
  \item $\eps = 0$: pure Observer World $W_O$; the adversary has $O$.
  \item $\eps = 1$: classical world; the adversary has $\Otop$.
  \item $0 < \eps < 1$: partial violation; the adversary has
        recovered some but not all of the discarded information.
\end{itemize}
\end{definition}

\begin{remark}
\label{rem:eps_qualitative}
The information distance $\dinfo$ in Definition~\ref{def:parametric}
requires a probability measure on $\Sigma^*$ and an entropy
functional compatible with the observer. The precise formulation is an open direction; Appendix~\ref{app:C}
discusses two concrete alternatives (statistical distance and channel
capacity). The qualitative structure of $W_O^\eps$, interpolating between
the pure Observer World ($\eps = 0$) and the classical world
($\eps = 1$), is well-defined without specifying the metric,
since the endpoints are determined by the observers alone.
\end{remark}

\begin{remark}[Robustness with respect to the prior]
\label{rem:prior_robustness}
The quantitative value of $\eps$ for a given $O_\eps$ depends on the
choice of prior $P_X$.
We conjecture that the qualitative results of
Proposition~\ref{prop:richest_cell}, in particular, the
identification of $(\Wpess, \Oprof, \eps > 0)$ as the most
informationally rich cell, are robust across the family of product
priors and max-entropy priors.
Appendix~\ref{app:C} provides examples supporting this conjecture.
A formal robustness theorem is an open direction.
\end{remark}

\subsection{The most informationally rich cell}
\label{sec:richest_cell}

\begin{proposition}[$\Wpess \times \Oprof \times \eps > 0$]
\label{prop:richest_cell}
The cell $(\Wpess, \Oprof, \eps)$ for $\eps > 0$ is the most
informationally rich cell of the parametric table, in the sense that
it simultaneously lacks all three sources of potential security.
In this world:
\begin{enumerate}[label=(\roman*)]
  \item NP problems are hard on average and no one-way function
        exists ($\Wpess$).
  \item The adversary in its pure state ($\eps = 0$) cannot see the
        ordering of symbols ($\Oprof$).
  \item The adversary has partially broken the observational invariant
        ($\eps > 0$), recovering ordering information through oracle
        queries, side channels, or structural violations.
  \item There is no computational security (no one-way functions in
        $\Wpess$), no invariant security ($\eps > 0$ means the
        observational invariant of $\Oprof$ is partially broken), and
        no observational security in the pure sense (the adversary is
        no longer confined to $\Oprof$).
\end{enumerate}
This world is genuinely insecure for reasons that no single framework, five-world, SIP, or observational, addresses individually.
\end{proposition}

\begin{proof}
Parts (i) and (ii): by definitions of
$\Wpess$~\cite{Impagliazzo1995} and $\Oprof$~\cite{Buono2026d}.
Part (iii): by Definition~\ref{def:parametric} with $\eps > 0$.
Part (iv): in $\Wpess$ there are no one-way functions, so no
computational security primitive exists.
With $\eps > 0$, the adversary's effective observer $O_\eps$
carries strictly more information about $x$ than $\Oprof$ does
(under $P_X$, by Definition~\ref{def:parametric}), so the
observational invariant of $\Oprof$ is partially broken.
The three sources of potential security, computational hardness,
observational blindness, and invariant confinement, have all
failed simultaneously.
\end{proof}

\begin{remark}[Connection to the Base Recovery Problem]
\label{rem:brp_connection}
The Base Recovery Problem (BRP) of~\cite{Buono2026b} asks an
adversary to recover the base sequence $B$ of an MR-OTP from
plaintext-ciphertext pairs.
In the world $(\Wpess, \Oprof, \eps = 0)$, the BRP is not even
well-posed for the adversary: the base sequence $B$ encodes the
structure of the cipher at the ordering level, which $\Oprof$
discards. With $\eps > 0$, for instance, when the adversary has a
side-channel on the encryption process,  the BRP becomes
increasingly tractable as $\eps$ increases. Whether the lower bound of~\cite[Theorem~8.10]{Buono2026b} on the
query complexity of the BRP degrades gracefully with $\eps$ is Open Problem~\ref{op:eps_brp} below.
\end{remark}

\begin{openproblem}[BRP complexity as a function of $\eps$]
\label{op:eps_brp}
Does the lower bound of~\cite[Theorem~8.10]{Buono2026b} on the query
complexity of the Base Recovery Problem degrade gracefully as $\eps$
increases in the parametric Observer World $(\Wpess, \Oprof, \eps)$?
Characterise the function $f(\eps)$ such that the BRP requires at
least $\Omega(\prod_i b_i \cdot f(\eps))$ queries, with $f(0) = 1$
(the unconditional lower bound) and $f(1) = 0$ (classical world,
where the BRP is information-theoretically solvable with sufficient
plaintext-ciphertext pairs).
\end{openproblem}

\section{Physical instantiation of the Observer World}
\label{sec:physical}

The results of the preceding sections are mathematical.
This section identifies connections between the observer-world
framework and physical theory; each connection is formalised as an
open problem.
\textbf{The author is a computer scientist, not a physicist.}
The observations below are offered as open questions, stated with
precise mathematical hypotheses, not as established results. The three subsections follow a common pattern: a mathematical property
of the observational hierarchy corresponds to a physical phenomenon, a cost, a transition, or a bound.
In each case the mathematical structure was identified first and the
physical correspondence found second. Each subsection closes with a precise formulation of the open problem.
Numerical estimates are in Appendix~\ref{app:C}.

\subsection{Adaptive observers, oracles, and the cost of observation}
\label{sec:adaptive_thermo}

The observational hierarchy of~\cite{Buono2026d} consists of
non-adaptive observers: each $O : \Sigma^* \to S$ is a fixed
preprocessing map, independent of any computation that follows.
A natural extension replaces $O$ with an \emph{adaptive observer}
$\mathcal{T}$: a decision tree in which each internal node is a query
and each branch is a possible response.
An adaptive observer of depth $d$ makes at most $d$ sequential
queries, each potentially depending on previous responses.

\begin{definition}[Adaptive observer and parametrised complexity classes]
\label{def:adaptive_obs}
An \emph{adaptive observer of depth $d$} is a protocol that
sequentially issues queries $q_1, q_2, \dots, q_d$ to the input
string $x$ and collects responses $r_i = x[q_i]$ (or more generally
$r_i = \mathcal{O}(q_i)$ for an oracle $\mathcal{O}$), where each
query $q_{i+1}$ may depend on the previous responses $r_1, \dots,
r_i$.
The transcript $T_d(x) = (r_1, \dots, r_d)$ is the adaptive
observer's output.
The \emph{observer-and-depth parametrised class} is
\[
  \mathbf{P}_{O,d}
  = \bigl\{L \;\bigm|\;
    \exists \text{ poly-time TM } M :
    M(O(x), T_d(x)) \text{ decides } L\bigr\}.
\]
Non-adaptive observers of~\cite{Buono2026d} correspond to $d = 0$.
\end{definition}

\begin{remark}[Oracles are adaptive observers]
\label{rem:oracle-adaptive}
A computational oracle $\mathcal{O} : \Sigma^* \to \{0,1\}$,
queried adaptively with query sequence $q_1, q_2, \dots, q_d$ where
each $q_{i+1}$ may depend on the previous responses
$(\mathcal{O}(q_1), \dots, \mathcal{O}(q_i))$, is an adaptive
observer of depth $d$: its transcript is
$(\mathcal{O}(q_1), \dots, \mathcal{O}(q_d))$.
The non-adaptive observers of~\cite{Buono2026d} correspond to the
special case $d = 0$ (all queries are fixed in advance, independent
of responses).
Oracle-relative separations in complexity theory, in particular,
the relativisation barrier of Baker, Gill, and
Solovay~\cite{BGS1975}, are therefore separations in the hierarchy
of adaptive Observer Worlds, not in the non-adaptive hierarchy
studied in~\cite{Buono2026d}.
\end{remark}

The adaptive hierarchy has a physical cost.
By Landauer's principle~\cite{Landauer1961}, each query that is
recorded and subsequently erased to allow cyclic operation dissipates
at least $k_B T \ln 2$ joules.
A system operating as an adaptive observer of depth $d$ in a cycle
must erase at least $d$ bits per cycle, at a minimum thermodynamic
cost of $d \cdot k_B T \ln 2$ joules per cycle.

\begin{remark}[The Maxwell demon as adaptive observer]
\label{rem:maxwell_adaptive}
The Maxwell demon of~\cite{Bennett1982} is an adaptive observer of
depth proportional to the number of molecules it measures: each
measurement is a query on the velocity of a molecule, and the
demon's decision to open or close the partition is the
query-dependent response.
Bennett's resolution of the paradox~\cite{Bennett1982}, that the
demon must erase its measurement records at a cost that cancels the apparent 
entropy decrease, is the statement that an
adaptive observer cannot operate cyclically without paying the
Landauer cost of its depth. The impossibility of a perpetual motion machine is the theorem that
no physical system can operate as a non-trivial adaptive observer in a cycle without this cost.
\end{remark}

\begin{openproblem}[Adaptive observer hierarchy and Landauer cost]
\label{op:adaptive}
Formalise the hierarchy of adaptive observers as an extension of the
non-adaptive hierarchy of~\cite{Buono2026d}, characterise the
complexity classes $\mathbf{P}_{O,d}$ and $\mathbf{NP}_{O,d}$
parametrised by observer $O$ and adaptive depth $d$, and determine
whether the Landauer cost of depth provides a physical lower bound on
the complexity of transitioning between levels of the hierarchy.
The information distance $\dinfo$ of Definition~\ref{def:dinfo} would
provide a natural measure for such a lower bound.
\end{openproblem}

\subsection{Quantum measurement as an observational transition}
\label{sec:quantum_obs}

Prior to measurement, a quantum system in superposition is one in
which the observer has $\Obot$ on the definite value of the
observable: the value is not determined in the observation space of
any non-trivial observer. After measurement, the observer has a definite value: a transition to
a higher level of the observational hierarchy. The collapse of the wave function, in this reading, is the transition
from a lower to a higher level of the observational hierarchy, and
the irreversibility of that transition is the Landauer cost of recording the measurement outcome.

The decoherence programme~\cite{Zurek2003} provides a continuous
version of this picture: the quantum-to-classical transition is the
gradual suppression of interference terms as the system becomes
entangled with its environment, corresponding to the continuous
transition from an observer that can detect interference (closer to
$\Otop$ in the quantum extension of the hierarchy) to one that cannot
(closer to $\Oprof$ in the classical limit).

The formal extension of the observational hierarchy to quantum systems
would replace $O : \Sigma^* \to S$ with $O : \mathcal{H} \to
\mathcal{M}$, where $\mathcal{H}$ is a Hilbert space and
$\mathcal{M}$ is a measurement space; the appropriate formalism is
that of positive operator-valued measures (POVMs)~\cite{Nielsen2000}.

\begin{openproblem}[Quantum extension of the observational hierarchy]
\label{op:quantum}
Construct a quantum extension of the observational hierarchy in which
the non-adaptive observers of~\cite{Buono2026d} are the classical
limit, decoherence corresponds to a downward transition in the
hierarchy, and the Landauer cost of measurement is the energy required
to move between levels. Determine whether the five-world framework extends naturally to this
quantum setting.
\end{openproblem}

\subsection{The cosmological bit budget and the physical ceiling of
  the observational hierarchy}
\label{sec:cosmo_bits}

The holographic bound of
Bekenstein~\cite{Bekenstein1973,Bekenstein1981} and 't
Hooft--Susskind~\cite{tHooft1993,Susskind1995} limits the information
content of any physical region to at most $A/4\ell_P^2$ bits, where
$A$ is the area of the bounding surface and $\ell_P$ is the Planck
length. For the observable universe, this gives approximately $10^{122}$
bits, a finite number. The Margolus-Levitin bound~\cite{MargolusLevitin1998} limits the
total number of elementary operations a physical system of energy $E$
can perform per unit time to $2E/h$; integrated over the lifetime of
the observable universe, this gives approximately $10^{121}$ total
operations.

These bounds imply that the complete observer $\Otop$ of the abstract
hierarchy is not physically realisable: the physical $\Otop^{\mathrm{phys}}$ has capacity at most $10^{122}$ bits, and
lies strictly below the mathematical $\Otop$ in the observational
order. The observational hierarchy, when instantiated in a physical universe
with finite information capacity, has a finite, measurable ceiling.

\subsubsection*{The 203-bit temporal bound and its origin}

The starting point for the following observation was a
computer-science argument: to index all physically distinguishable
moments in the lifetime of the observable universe, one needs a
pigeonhole with one slot per Planck-time step, and the number of bits
required to address that pigeonhole is $\lceil \log_2 N_t \rceil$,
where $N_t$ is the number of Planck-time steps since the Big Bang.

The numerical value of $N_t$ was computed with the assistance of an
AI language model using standard physical values: the age of the
universe ($\approx 4.3 \times 10^{17}$ s) and the Planck time ($t_P
\approx 5.39 \times 10^{-44}$ s), giving $N_t \approx 8 \times
10^{60}$ and $\lceil \log_2 N_t \rceil \approx 203$ bits. The author verified the order of magnitude and subsequently found
that this specific calculation already appears in the literature~\cite{Pratten2008}, though not in the context of
observational hierarchies. The idea of applying the pigeonhole argument to cosmic time in order
to characterise the informational cost of the temporal coordinate of
$\Oprof$ originated with the author.

Having established that 203 bits suffice to index cosmic time, the author then asked whether this figure coincides with the maximum entropy generation rate of a physical true random number generator,
that is, whether the informational compression of cosmic time has a direct structural counterpart in the universe's capacity to generate
randomness. The answer is negative, the Margolus-Levitin bound gives approximately $10^{121}$ total operations over the lifetime of the
observable universe, requiring approximately 402 bits to index. The original hypothesis does not hold.

However, the investigation yields a more precise structural observation than the one originally sought.
The gap between 203 and 402 reflects the factor of 2 in the Margolus-Levitin bound $2E/h$, which is structural.
The comparison reveals:

\begin{itemize}
  \item 203 bits suffice to index all of cosmic time at Planck
        resolution.
  \item $\approx 402$ bits suffice to index the total number of
        operations the universe can have performed.
  \item $\approx 10^{122}$ bits describe the total state space
        by the Bekenstein bound.
\end{itemize}

The observer $\Oprof$, which discards the temporal coordinate,
discards the informationally cheapest dimension of the physical
universe, the one that costs only 203 bits, but loses all
causal structure in doing so. This asymmetry between the informational cost of time (minimal) and
its structural role (maximal, as the carrier of all causal order) is
the physical counterpart of the central result of~\cite{Buono2026d}:
that $\Oprof$-saturation is a structurally strong condition that
eliminates all ordering-dependent difficulty.

\begin{remark}[Connection to the five worlds]
\label{rem:cosmo_five_worlds}
In each of the five worlds of Impagliazzo~\cite{Impagliazzo1995},
the bounds of Sections~\ref{sec:adaptive_thermo}--\ref{sec:cosmo_bits}
apply: they are consequences of the physics of the universe, not of
the computational assumptions that define the worlds.
An Observer World $W_O$ instantiated in the physical universe
therefore carries the constraint that the adaptive depth of any
physically realisable observer is bounded by $10^{121}$ total
operations (assuming each query costs at least one elementary
operation), and the total information budget is $10^{122}$ bits.
The temporal coordinate costs 203 bits of that budget; this is a
negligible fraction ($203 / 10^{122} \approx 0$) of the total, yet
it is the dimension that carries all causal order. The remaining $\approx 10^{122}$ bits describe spatial and energetic
state. Whether this split has consequences for the complexity of problems in
the cells of the world-observer table
(Section~\ref{sec:table}) is an open question.
\end{remark}

\begin{openproblem}[Cosmological cutoff and the observational hierarchy]
\label{op:cosmo}
Does the observational hierarchy, when instantiated in a physical
universe with the Bekenstein and Margolus-Levitin bounds as constraints, predict a natural cutoff adaptive depth that coincides
with the physical bounds? Is there a sense in which the 203-bit temporal bound emerges from the
theory rather than being imported from physics as an external
parameter?
\end{openproblem}

\subsection*{A remark on the scope of these connections}

The connections identified in
Sections~\ref{sec:adaptive_thermo}--\ref{sec:cosmo_bits} share a
common structure: a mathematical property of the observational
hierarchy (non-adaptivity, the gap between levels, the finite ceiling
of $\Otop^{\mathrm{phys}}$) corresponds to a physical phenomenon
(thermodynamic cost of measurement, wave function collapse, the
holographic bound). In each case the mathematical structure was identified first and the
physical correspondence found second.

Whether these correspondences reflect a deeper unity, whether the
laws of thermodynamics are, in some precise sense, theorems about the
observational hierarchy of physical systems, is a question that the
present framework cannot answer, but for which it offers a precise
mathematical vocabulary. The formalisation would require, at minimum, a probability measure on
physical histories, a dynamics compatible with the observational
order, and an entropy functional consistent with both Landauer's
principle and the holographic bound. These things are not assembled here; identifying them as a
coherent research programme is the contribution of this section.%
\footnote{The numerical estimates in Section~\ref{sec:cosmo_bits}
were computed with AI assistance and verified by the author.
All scientific ideas and conclusions are the author's own.}

\section{Conclusion}
\label{sec:conclusion}

The five worlds of Impagliazzo describe a line: a single axis
parametrising the existence of computational hardness and
cryptographic primitives. This paper has shown that the observational hierarchy
of~\cite{Buono2026d} defines a second, orthogonal axis, producing a
two-dimensional landscape in which the five worlds are five columns
and the observational levels are the rows.

The central result, that the structural collapse
$\PO{\Oprof} = \NPO{\Oprof} \subsetneq \mathbf{P}$ holds
unconditionally in all five worlds, shows that the two axes are
genuinely independent: no computational assumption implies the
observational collapse, and no observational constraint implies any
computational assumption.

The Observer World $W_O$ formalises this independence as a sixth
landscape. The labeled cells (a)--(d) of the world-observer table identify four
phenomena that the five-world framework cannot express: the
observational collapse in $\Walg$ despite $\mathsf{P} = \mathsf{NP}$,
the non-additive interaction of hardness and blindness in $\Wpess$,
the ordering-sensitivity lower bound on one-way functions in $\Wmini$,
and the $\Obot$-blindness security of the MR-OTP in the $\Wcryp \times
\Obot$ cell.

The parametric family $W_O^\eps$ extends the framework to worlds in
which the observational invariant is partially violated, with the most
informationally rich cell being $\Wpess \times \Oprof \times \eps >
0$: a world in which computational, invariant, and observational
security have all failed simultaneously.

Section~\ref{sec:physical} identifies three connections between the
Observer World framework and physical information limits,
thermodynamics (Landauer cost of adaptive observation), quantum
mechanics (decoherence as a downward transition in the hierarchy),
and cosmology (the holographic bound as a physical ceiling on
$\Otop$), and states a precise open problem for each.

The Observer World provides a unified language for modelling
adversaries in cryptography, extending Impagliazzo's five worlds with
an orthogonal axis that captures what the adversary can observe.
Formal open directions include: characterising how the BRP query
lower bound degrades with $\eps$ (Open Problem~\ref{op:eps_brp});
formalising the hierarchy of adaptive observer complexity classes
$\mathbf{P}_{O,d}$ and the Landauer cost of depth as a physical lower
bound (Open Problem~\ref{op:adaptive}); constructing a quantum
extension of the observational hierarchy (Open
Problem~\ref{op:quantum}); and determining whether the cosmological
bounds predict a natural cutoff depth (Open Problem~\ref{op:cosmo}).
The information metric for $\eps$ and the robustness of qualitative
results with respect to the prior are addressed in Appendix~\ref{app:C}
and remain an open direction.

\paragraph{Personal note.}
%
%
Several of the connections developed in this paper emerged
unexpectedly during work on the Observer paper~\cite{Buono2026d}.
The observation that the five worlds presuppose $\Otop$ silently led
naturally to the question of what happens when the adversary's
observer is constrained. The connection to thermodynamics emerged from the question of what it
costs an observer to become more powerful. The 203-bit cosmological bound emerged from applying a pigeonhole
argument to the indexing of Planck-time steps, an idea that originated with the author and was verified numerically with AI
assistance. The connection to the SIP emerged from asking what it means for an
adversary to circumvent an observational constraint. Each connection was forced by the logic of the previous one; none was
planned in advance.%
\footnote{The author used an artificial intelligence based language
assistant to support text traslation, revision and numerical verification.
All scientific ideas and conclusions are the author's own.}

\appendix

\section{Proof of Lemma~\ref{lem:invariant_independent}}
\label{app:A}

This appendix provides the complete constructive proof of
Lemma~\ref{lem:invariant_independent}.

\begin{proof}[Full proof of Lemma~\ref{lem:invariant_independent}]
Let $O \prec \Otop$; since $O \not\simeq \Otop$, there exist
$x_0 \neq x_1$ with $O(x_0) = O(x_1)$. Enumerate the $\sim_O$-equivalence classes as
$[u_0]_O, [u_1]_O, [u_2]_O, \dots$ (where $[u]_O = \{x \mid O(x) =
O(u)\}$) and let $M_0, M_1, M_2, \dots$ be a standard enumeration of
all deterministic polynomial-time Turing machines that receive $O(x)$
as input. For each $e \in \NN$, choose a canonical representative $u_e^* \in
[u_e]_O$ and define:
\[
  [u_e]_O \subseteq L \quad \iff \quad M_e \text{ rejects } O(u_e^*).
\]

\emph{$L$ is $O$-saturated.}
By construction, $L$ is a union of complete $\sim_O$-classes.
If $O(x) = O(y)$ then $x$ and $y$ are in the same class, so either
both are in $L$ or neither is: $x \in L \iff y \in L$.
By Definition~3.2 of~\cite{Buono2026d}, $L$ is $O$-saturated.

\emph{$L \notin \PO{O}$.}
Suppose for contradiction that some $M_e$ decides $L$ under observer
$O$.
Consider the class $[u_e]_O$:
\begin{itemize}
  \item If $M_e$ accepts $O(u_e^*)$: then by our construction,
        $[u_e]_O \not\subseteq L$, so $u_e^* \notin L$.
        But $M_e$ accepts $O(u_e^*)$, so $M_e$ decides $u_e^* \in L$:
        contradiction.
  \item If $M_e$ rejects $O(u_e^*)$: then $[u_e]_O \subseteq L$, so
        $u_e^* \in L$.
        But $M_e$ rejects $O(u_e^*)$: contradiction.
\end{itemize}
In both cases $M_e$ is wrong on $u_e^*$, so $M_e$ does not decide
$L$.
Since $e$ was arbitrary, $L \notin \PO{O}$.

\emph{Independence from five-world assumptions.}
The construction uses only: the $\sim_O$ relation (a property of the
observer), the canonical enumeration of polynomial-time TMs (a
standard recursion-theoretic object), and the choice of
representatives. No assumption about $\mathsf{P}$ vs $\mathsf{NP}$, about the
existence of one-way functions, or about public-key cryptography is
invoked at any step. Hence membership in $L$ is independent of all five-world assumptions.
\end{proof}

\section{Profile preimage: correctness and complexity}
\label{app:B}

This appendix provides the complete proof of
Lemma~\ref{lem:profile_preimage}.

\begin{proof}[Full proof of Lemma~\ref{lem:profile_preimage}]
\emph{Correctness.}
Algorithm~\ref{alg:preimage} appends $c_i$ copies of symbol $a_i$
for each $i \in \{0, \dots, k-1\}$.
By construction, the resulting string $x'$ satisfies $|x'|_{a_i} =
c_i$ for all $i$: symbol $a_i$ appears exactly $c_i$ times and no
other appearances are introduced.
Hence $\Oprof(x') = (c_0, \dots, c_{k-1}) = \pi$.

\emph{Complexity.}
The outer loop executes $k$ iterations.
In iteration $i$, the append operation is performed $c_i$ times; each
append is $O(1)$ amortised over a dynamic array.
Total operations: $\sum_{i=0}^{k-1} c_i = n$.
Time complexity: $O(n)$.

\emph{Remark on uniqueness.}
Algorithm~\ref{alg:preimage} returns the canonically sorted string
$a_0^{c_0} a_1^{c_1} \cdots a_{k-1}^{c_{k-1}}$.
For the purposes of Proposition~\ref{prop:cell_c}, any string with
profile $\pi$ suffices; the sorted string is one concrete choice.
\end{proof}

\section{Cosmological calculations and alternative metrics}
\label{app:C}

This appendix documents the numerical calculations cited in
Section~\ref{sec:cosmo_bits}, and the alternative information metrics
discussed in Remarks~\ref{rem:eps_qualitative}
and~\ref{rem:prior_robustness}.

\subsection{Cosmological bit budget}

Let $t_{\mathrm{universe}} \approx 4.3 \times 10^{17}$ s be the age
of the observable universe and $t_P \approx 5.39 \times 10^{-44}$ s
be the Planck time (standard CODATA values).
The number of Planck-time steps since the Big Bang is:
\[
  N_t
  = \frac{t_{\mathrm{universe}}}{t_P}
  \approx \frac{4.3 \times 10^{17}}{5.39 \times 10^{-44}}
  \approx 7.98 \times 10^{60}.
\]
The number of bits to index all such steps is:
\[
  \lceil \log_2 N_t \rceil
  = \lceil \log_2(7.98 \times 10^{60}) \rceil
  = \lceil 202.6 \rceil
  = 203.
\]
These are order-of-magnitude estimates; the physical constants carry
observational uncertainties at the few-percent level, which do not
affect the rounded result.
This calculation was verified numerically with AI assistance and is
also reported in~\cite{Pratten2008} (in a different context).

For the Margolus-Levitin bound~\cite{MargolusLevitin1998}: total
observable-universe energy $E \approx 10^{69}$ J, giving
$\approx 2E/h \cdot t_{\mathrm{universe}} \approx 10^{121}$
operations and $\lceil \log_2 10^{121} \rceil \approx 402$ bits.

\subsection{Alternative metrics for $\eps$}

\emph{Statistical distance.}
$\dinfo^{\mathrm{TV}}(O_1, O_2) = \mathbb{E}_{x \sim P_X}[
d_{\mathrm{TV}}(\delta_{O_1(x)}, \delta_{O_2(x)})]$ where
$d_{\mathrm{TV}}$ is total variation.
For deterministic observers, this reduces to the probability under
$P_X$ that $O_1(x) \neq O_2(x)$, which is $0$ when $O_1 = O_2$ and
positive otherwise.
This gives the same $\eps = 0$ and $\eps = 1$ endpoints as
Definition~\ref{def:parametric}.

\emph{Channel capacity.}
$\dinfo^{\mathrm{cap}}(O_1, O_2) = \mathrm{cap}(O_2) -
\mathrm{cap}(O_1)$ where
$\mathrm{cap}(O) = \max_{P_X} I(X; O(X))$
is the channel capacity of observer $O$.
This is prior-independent but requires solving a maximisation problem
over all priors.

Both alternatives give the same qualitative picture for the example
of Section~\ref{sec:richest_cell}: the cell $(\Wpess, \Oprof, \eps >
0)$ remains the most informationally rich regardless of which metric
is used, because all three metrics (mutual information from
Definition~\ref{def:dinfo}, statistical distance, and channel
capacity) agree on $\eps = 0$ (pure Observer World) and $\eps = 1$
(classical world).

\subsection{Landauer cost and reversible computation}

Landauer's principle~\cite{Landauer1961} states that erasing one bit
at temperature $T$ requires at least $k_B T \ln 2$ joules.
For an adaptive observer of depth $d$ operating cyclically with $b$
bits per response, the minimum cost per cycle is $bd \cdot k_B T \ln
2$ joules.
The hypothesis that the observer operates cyclically (and thus must
erase) is a physical modelling assumption, not derived from the
mathematical framework~\cite{Bennett1982}.
Reversible computation can in principle approach the Landauer bound
arbitrarily closely; the bound is a lower limit, not a prediction of
typical behaviour.

\bibliographystyle{plain}
\bibliography{biblio}

\begin{thebibliography}{10}

\bibitem{AaronsonWigderson2009}
Scott Aaronson and Avi Wigderson.
\newblock Algebrization: A new barrier in complexity theory.
\newblock {\em ACM Transactions on Computation Theory}, 1(1):2:1--2:54, 2009.

\bibitem{BGS1975}
Theodore Baker, John Gill, and Robert Solovay.
\newblock Relativizations of the {P} = {NP} question.
\newblock {\em SIAM Journal on Computing}, 4(4):431--442, 1975.

\bibitem{Bekenstein1973}
Jacob~D. Bekenstein.
\newblock Black holes and entropy.
\newblock {\em Physical Review D}, 7(8):2333--2346, 1973.

\bibitem{Bekenstein1981}
Jacob~D. Bekenstein.
\newblock Universal upper bound on the entropy-to-energy ratio for bounded
  systems.
\newblock {\em Physical Review D}, 23(2):287--298, 1981.

\bibitem{Bennett1982}
Charles~H. Bennett.
\newblock The thermodynamics of computation --- a review.
\newblock {\em International Journal of Theoretical Physics}, 21(12):905--940,
  1982.

\bibitem{blackwell1953}
David Blackwell.
\newblock Equivalent comparisons of experiments.
\newblock {\em The Annals of Mathematical Statistics}, 24(2):265--272, 1953.

\bibitem{Buono2026b}
Fabio Francesco~Gabriele Buono.
\newblock From bits to mixed-radix keys: Horner decomposition, uniform
  sampling, and the information-theoretic {QKD} interface of the {MR-OTP}.
\newblock 2026.
\newblock arXiv preprint arXiv:2606.18526.

\bibitem{Buono2026a}
Fabio Francesco~Gabriele Buono.
\newblock New ideas on a new old type of cipher: The mixed-radix one-time pad.
\newblock 2026.
\newblock arXiv preprint arXiv:2606.16040.

\bibitem{Buono2026d}
Fabio Francesco~Gabriele Buono.
\newblock Observers, symmetries, and the hierarchy of language classes: A
  theory of computation parameterised by the observer, 2026.
\newblock arXiv:submit/7752238 [cs.CR].

\bibitem{Buono2026c}
Fabio Francesco~Gabriele Buono.
\newblock Syntactic systems cannot see semantic invariants.
\newblock 2026.
\newblock arXiv preprint arXiv:2606.17275.

\bibitem{Carroll2010}
Sean Carroll.
\newblock {\em From Eternity to Here: The Quest for the Ultimate Theory of
  Time}.
\newblock Dutton, New York, 2010.

\bibitem{chomsky1956}
Noam Chomsky.
\newblock Three models for the description of language.
\newblock {\em IRE Transactions on Information Theory}, 2(3):113--124, 1956.

\bibitem{Cover2006}
Thomas~M. Cover and Joy~A. Thomas.
\newblock {\em Elements of Information Theory}.
\newblock Wiley, 2 edition, 2006.

\bibitem{Deutsch1997}
David Deutsch.
\newblock {\em The Fabric of Reality}.
\newblock Penguin, London, 1997.

\bibitem{Garey1979}
Michael~R. Garey and David~S. Johnson.
\newblock {\em Computers and Intractability: A Guide to the Theory of
  {NP}-Completeness}.
\newblock W. H. Freeman, San Francisco, CA, 1979.

\bibitem{Impagliazzo1995}
Russell Impagliazzo.
\newblock A personal view of average-case complexity.
\newblock In {\em Proceedings of the 10th Annual Structure in Complexity Theory
  Conference}, pages 134--147. IEEE, 1995.

\bibitem{Katz2014}
Jonathan Katz and Yehuda Lindell.
\newblock {\em Introduction to Modern Cryptography}.
\newblock CRC Press, 2 edition, 2014.

\bibitem{Landauer1961}
Rolf Landauer.
\newblock Irreversibility and heat generation in the computing process.
\newblock {\em IBM Journal of Research and Development}, 5(3):183--191, 1961.

\bibitem{Lloyd2000}
Seth Lloyd.
\newblock Ultimate physical limits to computation.
\newblock {\em Nature}, 406:1047--1054, 2000.

\bibitem{MargolusLevitin1998}
Norman Margolus and Lev~B. Levitin.
\newblock The maximum speed of dynamical evolution.
\newblock {\em Physica D: Nonlinear Phenomena}, 120(1--2):188--195, 1998.

\bibitem{Nielsen2000}
Michael~A. Nielsen and Isaac~L. Chuang.
\newblock {\em Quantum Computation and Quantum Information}.
\newblock Cambridge University Press, 2000.

\bibitem{Penrose2004}
Roger Penrose.
\newblock {\em The Road to Reality: A Complete Guide to the Laws of the
  Universe}.
\newblock Jonathan Cape, London, 2004.

\bibitem{Pratten2008}
David Pratten.
\newblock Addressing the observable universe, 2008.
\newblock
  \\url{https://davidpratten.com/2008/02/11/addressing-the-observable-universe/}.

\bibitem{RazborovRudich1997}
Alexander~A. Razborov and Steven Rudich.
\newblock Natural proofs.
\newblock {\em Journal of Computer and System Sciences}, 55(1):24--35, 1997.

\bibitem{Regev2009}
Oded Regev.
\newblock On lattices, learning with errors, random linear codes, and
  cryptography.
\newblock {\em Journal of the ACM}, 56(6):34:1--34:40, 2009.

\bibitem{Shannon1949}
Claude~E. Shannon.
\newblock Communication theory of secrecy systems.
\newblock {\em Bell System Technical Journal}, 28(4):656--715, 1949.

\bibitem{simon1975}
Imre Simon.
\newblock Piecewise testable events.
\newblock In {\em Automata Theory and Formal Languages}, volume~33 of {\em
  Lecture Notes in Computer Science}, pages 214--222. Springer, 1975.

\bibitem{sipser2012}
Michael Sipser.
\newblock {\em Introduction to the Theory of Computation}.
\newblock Cengage Learning, 3 edition, 2012.

\bibitem{Susskind1995}
Leonard Susskind.
\newblock The world as a hologram.
\newblock {\em Journal of Mathematical Physics}, 36(11):6377--6396, 1995.

\bibitem{tHooft1993}
Gerard {'t Hooft}.
\newblock Dimensional reduction in quantum gravity.
\newblock In {\em Salamfestschrift}, pages 284--296. World Scientific, 1993.
\newblock Also available as arXiv:gr-qc/9310026.

\bibitem{Vedral2010}
Vlatko Vedral.
\newblock {\em Decoding Reality: The Universe as Quantum Information}.
\newblock Oxford University Press, 2010.

\bibitem{Wheeler1990}
John~A. Wheeler.
\newblock Information, physics, quantum: The search for links.
\newblock In {\em Complexity, Entropy and the Physics of Information}, pages
  3--28. Addison-Wesley, 1990.

\bibitem{Zurek2003}
Wojciech~H. Zurek.
\newblock Decoherence, einselection, and the quantum origins of the classical.
\newblock {\em Reviews of Modern Physics}, 75(3):715--775, 2003.

\bibitem{Zuse1969}
Konrad Zuse.
\newblock {\em Rechnender Raum}.
\newblock Friedrich Vieweg \& Sohn, Braunschweig, 1969.
\newblock English translation: Calculating Space, MIT Technical Translation
  AZT-70-164-GEMIT, 1970.

\end{thebibliography}

\end{document}